\pdfoutput=1
\documentclass[a4paper]{lipics-v2021}
\hideLIPIcs

\usepackage{xcolor, soul}
\usepackage[noline,noend]{algorithm2e}
\usepackage{relsize}
\usepackage{listings}
\usepackage{verbatim}
\usepackage{graphicx}
\usepackage{subcaption}
\usepackage[subtle]{savetrees}
\usepackage{wrapfig}
\usepackage{float}

\usepackage{pgfplots}
\usepackage{pgfplotstable}
\usetikzlibrary{pgfplots.groupplots}

\usepackage{tikz}
\usetikzlibrary{automata, positioning, arrows}

\pgfplotsset{compat=1.17}

\AtBeginDocument{%
  \providecommand\BibTeX{{%
    \normalfont B\kern-0.5em{\scshape i\kern-0.25em b}\kern-0.8em\TeX}}}

\newcommand{\sys}[0]{\mathbb{S}}
\newcommand{\introparagraph}[1]{\vspace{0.7mm} \noindent \textbf{\em #1.}}

\newif\ifcomments
\ifcomments
    \providecommand{\shadaj}[1]{{\protect\color{brown}{\bf [shadaj: #1]}}}
    \providecommand{\conor}[1]{{\protect\color{red}{\bf [conor: #1]}}}
    \providecommand{\alvin}[1]{{\protect\color{purple}{\bf [alvin: #1]}}}
    \providecommand{\paris}[1]{{\protect\color{blue}{\bf [paris: #1]}}}
    \providecommand{\joe}[1]{{\protect\color{teal}{\bf [joe: #1]}}}
    \providecommand{\jmh}[1]{{\protect\color{teal}{\bf [joe: #1]}}}
    \providecommand{\david}[1]{{\protect\color{green}{\bf [david: #1]}}}
    \providecommand{\chris}[1]{{\protect\color{violet}{\bf [chris: #1]}}}
    \providecommand{\davidmwei}[1]{{\protect\color{pink}{\bf [david wei: #1]}}}
    \providecommand{\val}[1]{{\protect\color{orange}{\bf [kaushik: #1]}}}
    \providecommand{\justin}[1]{{\protect\color{green}{\bf [justin: #1]}}}
    \providecommand{\mingwei}[1]{{\protect\color{rhodamine}{\bf [mingwei: #1]}}}
    \providecommand{\rithvik}[1]{{\protect\color{red}{\bf [rithvik: #1]}}}
    \providecommand{\nc}[1]{{\protect\color{pink}{\bf [nc: #1]}}}
     \providecommand{\accheng}[1]{{\protect\color{olive}{\bf [accheng: #1]}}}
    \providecommand{\dan}[1]{{\protect\color{purple}{\bf [dan: #1]}}}
  \else
    \providecommand{\shadaj}[1]{}
    \providecommand{\conor}[1]{}
    \providecommand{\alvin}[1]{}
    \providecommand{\paris}[1]{}
    \providecommand{\joe}[1]{}
    \providecommand{\jmh}[1]{}
    \providecommand{\david}[1]{}
    \providecommand{\chris}[1]{}
    \providecommand{\davidmwei}[1]{}
    \providecommand{\val}[1]{}
    \providecommand{\justin}[1]{}
    \providecommand{\mingwei}[1]{}
    \providecommand{\rithvik}[1]{}
    \providecommand{\nc}[1]{}
    \providecommand{\accheng}[1]{}
    \providecommand{\dan}[1]{}
\fi

\title{The Free Termination Property of Queries Over Time}

\author{Conor Power}{University of California Berkeley, United States}{conorpower@cs.berkeley.edu}{https://orcid.org/0000-0002-0660-5110}{}

\author{Paraschos Koutris}{University of Wisconsin Madison, United States}{paris@cswisc.edu}{https://orcid.org/0000-0001-6309-1702}{}

\author{Joseph M. Hellerstein}{University of California Berkeley, United States}{hellerstein@cs.berkeley.edu}{https://orcid.org/0000-0002-7712-4306}{}

\authorrunning{C. Power, P. Koutris, and J. Hellerstein}
    


\begin{document}
\nolinenumbers
\newpage
\setcounter{page}{1} 

\maketitle

\begin{abstract}
Building on prior work on distributed databases and the CALM Theorem, we define and study the question of \emph{free termination}: in the absence of distributed coordination, what query properties allow nodes in a distributed (database) system to unilaterally terminate execution even though they may receive additional data or messages in the future? This completeness question is complementary to the soundness questions studied in the CALM literature. We also develop a new model based on semiautomata that allows us to bridge from the relational transducer model of the CALM papers to algebraic models that are popular among software engineers (e.g. CRDTs) and of increasing interest to database theory for datalog extensions and incremental view maintenance.
\end{abstract}

\section{Introduction}
Distributed data systems have become a central topic in computing over the last decade due to a confluence of factors including cloud computing, increasing data growth and massively popular global-scale services. A central technical challenge in distributed computing and databases is the use---and avoidance---of coordination mechanisms~\cite{brewer2012cap,calmcacm}. Coordination in distributed systems is both slow and susceptible to unavailability via the CAP Theorem. 
In response, theoreticians have studied the question of what programs are computable in a distributed fashion without the use of coordination,
most notably in the CALM Theorem~\cite{calmTheorem,weakerMonotonicity, zinn-win-move, tim2023} for coordination-free queries.
Meanwhile, language designers and systems researchers have begun building systems that encourage coordination-free programming, exemplified by Conflict-free Replicated Data Types (CRDTs)~\cite{crdtOverview, crdtsSurvey}, which are popular data structure libraries based on semi-lattices.

In this work we push beyond CALM in two directions. The first is to consider a different proof goal. The work on CALM is fundamentally about a soundness property called \emph{coordination-free eventual consistency}: in the absence of coordination, what properties of a program ensure that each node in a distributed system will emit only correct program outputs over time? Here we pose a complementary completeness question, \emph{free termination}: in the absence of coordination, what properties of a program ensure that each node in a distributed system can unilaterally terminate after producing all its (correct) results, even if updates may arrive from other nodes in the future? From a practical perspective, free termination is critical to any user or client that requires a complete answer before proceeding.

The second ambition of this paper is to generalize the theory in this domain from its roots in relational transducers as explored in the CALM papers, and extend it to the context of the algebraic frameworks that are native to CRDTs, and of increasing interest in database research~\cite{datalogo, datalog-in-wonderland,dbsp,rings-ivm-koch, dbtoaster}. To that end, we introduce a general model that can capture what guarantees can be offered without coordination in both settings. Our model is based on queries over \emph{semiautomata} and the guarantees to users are captured as properties of those semiautomata.

The main contributions of the paper are as follows:
\begin{enumerate}
    \item Using semiautomata, we introduce the notion of free termination (Section~\ref{sec:free:termination}) in a state transition system, and show how it can be used to model different types of applications, including incremental view maintenance~\cite{dbsp} and the pre-semiring data model used to extend Datalog~\cite{datalogo}.
    \item We then explore how the algebraic properties of the system and the query affect free termination (Section~\ref{sec:partial-orders}). Among our results, we show that under acyclic state modifications (commonly found in CRDTs), the only queries with free termination are a particular class of threshold queries over the natural partial order of states. We also show that if updates form a group or ring (e.g., in incremental view maintenance models), free termination cannot be achieved.
    \item We study how to model coordination-free query computation in a distributed setting via the lens of free termination (Section~\ref{sec:distributed-systems}). Interestingly, we show that by using the notion of free termination we can achieve a stronger and more fine-grained notion of coordination-freeness that applies to a pair of a query and an input. We show how coordination-freeness for network transducers as defined in~\cite{calmTheorem} is a specific case of our more general notion of coordination-freeness. This allows us to characterize other queries as coordination-free, for example, antitone queries.
    \item Finally, we look into free termination when the state space is finite (Section~\ref{sec:finite}). We give a linear-time algorithm for deciding all free termination states in a transition system and also study how to perform state minimization. 
\end{enumerate}

\subsection{Motivating Examples}
Before we proceed, we provide some motivating examples from the literature. We start with the domain of CRDTs, which have become quite popular with software engineers and shed light on both of our goals.

\textbf{Grow-Only Set CRDT:} A common CRDT is the ``grow-only set'' CRDT: this is a replicated distributed set where each machine propagates its local set to other machines in the system. Upon receiving a message containing the local set at another machine, the local machine will apply its ``merge'' operation, modifying its local state to be the union of its current state and the incoming set.
This process of propagation over an asynchronous network introduces non-determinism as messages might be delivered in a different order than they are sent and messages may arrive multiple times. The fact that set union forms a semi-lattice ensures that, regardless of these sources of network nondeterminism, the state at each machine will eventually converge to the same value. That eventual value is the union of all of the initial states at each machine and the time at which it is reached is called the \textit{quiescence point}. 

Grow-only sets illustrate the point that CRDTs provide coordination-free consistency, but do not support free termination.
In the absence of coordination, we do not have a mechanism for determining locally whether we have received all the elements in the network---i.e. whether we have reached a quiescence point. Without a guarantee of quiescence, what should we do to answer a query from a user? CRDTs allow arbitrary queries at any time and make no guarantee on the relationship between the query result at time $t$ and the query result at the quiescence point. For example, consider the query $Q(x) = R(x) - T(x)$. A value $(a)$ may be returned at time $t$, but eventually be excluded from the final output via subsequent ``merges'' into $T(a)$. This is not a particularly satisfactory contract between the system and the user: what good is distributed state if you do not know when you can query it reliably?

\textbf{Threshold Queries:}
Threshold queries characterize a class of queries over CRDTs in which a machine is able to unilaterally detect that the output will never change regardless of future applications of the merge operation (set union in our example). For example, consider the query $Q() = |R(x)| > 10$. This boolean query is monotone with respect to the partial order of our semilattice (which is ordered by subset containment). Since the merge operation can only increase the position in the partial order, once the local state contains more than ten tuples in $R$, the result of the query is guaranteed to be the same at the quiescence time. By contrast, if the quiescence-time answer to this query is false---i.e. $\neg(|R(x)| > 10)$---then no machine will ever return an answer to this query. From a local perspective, a machine cannot know if there is some additional element out there that it has not heard about and will someday need to union into its local state. Threshold queries, while useful, only offer free termination for database instances where the answer is true! 

\textbf{CALM and Relational Transducers:}
The CALM Theorem~\cite{calmTheorem} uses relational transducer networks~\cite{relationalTransducers} to prove the relationship between queries expressed in monotonic logic and coordination-freeness. This formalism allows for the expression of distributed programs in terms of logical formulas that are evaluated iteratively on each machine and communicate state between machines. Similar models have been used by developers to implement distributed programs (e.g., Webdamlog~\cite{relationalTransducers}, 
Bloom~\cite{bloom}). While both the CALM theorem and CRDTs use proofs based on notions of monotonicity, there remains a gap between the logic formalism of the CALM theorem and the state-based formalisms of other work on coordination-freeness. 
In this paper, we are not interested in the programmability differences between these approaches
but exclusively in reasoning about guarantees on query results without coordination. 
In departing from the realm of logic programming, we want to reason naturally about programs beyond the boolean or relational setting. For example, consider a system that takes on integer values and supports the multiplication operation by a user-specified input integer. Now consider the query "is the running product divisible by two?". This query is not monotone with respect to the traditional ordering of the integers, but it can be computed coordination-free on input instances that contain any even number under the free-termination framework we introduce throughout this paper. Once we have multiplied by an even number, all possible future states will return True on this query.

As mentioned above, transducers were used to prove coordination-free consistency guarantees, but they do not directly address termination. As an example, consider a network of transducers supporting the Datalog language, and a simple program like transitive closure. Transducers accumulate knowledge about the set of paths in the global input, and once they learn about a path they can output it with certainty, but they will never conclude that they have finished finding new paths. Indeed, this depends on the database instance! We will return to this point in Section ~\ref{sec:distributed-systems}.

\section{Related Work}
The theory of what is computable without coordination in a distributed setting has been studied across different research communities and different theoretical models. In database theory, this has been studied in terms of the CALM Theorem ("Consistency as Logical Monotonicity")~\cite{declarativeImperative, calmTheorem, weakerMonotonicity, zinn-win-move, tim2023} using the relational transducer model of computation~\cite{relationalTransducers}. 

In related work in the programming languages community, 
conflict-free replicated data types (CRDTs) ~\cite{crdts} achieve coordination-free programs via the algebraic properties of functions that can modify the state of data.
CRDTs have found popularity amongst practicing software developers and are used in a variety of production software systems including Redis, Riak, ElectricSQL, SQLSync, Ditto, JupyterLab, and SoundCloud. 

CRDTs have been criticized for the guarantees that they offer~\cite{crdtOn}. 
To resolve this gap, several systems have combined semilattice state convergence with monotonic queries or functions over that state, including Bloom$^L$~\cite{bloomL}, Lasp~\cite{lasp, miller2017}, Datafun~\cite{datafun} and Hydroflow~\cite{hydro-initial-steps}.

For coordination-free termination detection,
LVars~\cite{lvars} was early in proposing the use of monotone threshold functions from any semilattice $L$ to a smaller lattice like the booleans $(\mathbb{B}, \vee)$. If a computation exceeds a threshold in $L$, the threshold function evaluates to \texttt{true}. Since \texttt{true} is the top (supremum) of $(\mathbb{B}, \vee)$, nothing can change the result, so computation can safely terminate. This idea can be used in the languages mentioned above.

Efforts at adding non-monotonic functions or queries to these languages have typically fallen back to the use of coordination.
Languages like Gallifrey~\cite{gallifrey} and Bloom$^L$ only guarantee consistency when non-monotone constructs are preceded by a round of coordination.

Beyond semilattices, recent research has explored alternative algebraic structures for data systems that each offer their own potential optimizations. 
Pre-semirings have been shown to offer semi-naive fixpoint evaluation in Datalog$^o$~\cite{datalogo, datalog-in-wonderland}. Abelian groups~\cite{dbsp} and rings~\cite{rings-ivm-koch, dbtoaster} have been shown to enable efficient incremental computation of materialized views. While each of these algebraic structures has shown promise in database contexts, one cannot always make use of them simultaneously (Section~\ref{sec:groups}).

In this work, we formalize the \textit{free termination} property for certainty of query answers, generalize it beyond the case of monotonic state with threshold queries, and show directly how this property manifests in non-CRDT settings like Datalog, relational transducer networks, and algebraic models of incremental view maintenance.

\section{Definitions of Free Termination}
\label{sec:free:termination}

We will capture different scenarios in distributed computation using the general computational model of a {\em semiautomaton}~\cite{algebraicAutomataBoook}. We assume that each node keeps a {\em state} that is represented by an element $s$ in a (finite or infinite) state space $D$. Computation at each node is modeled by a modification of the state $s \in D$: the {\em transition (or update) function} $U: D \times L \rightarrow D$ takes a state $s$, a parameter $\ell \in L$ from some domain $L$ and outputs a state that $s$ can transition to. We will often write $s \xrightarrow{\ell} s'$ to denote that $U(s,\ell) = s'$ (we will often omit $\ell$ if it is not of importance). 

\begin{definition}[Semiautomaton]
A {\em semiautomaton} is a triple $\sys = (D,L,U)$, where $D$ is a set called the state space, $L$ is a set called the parameter space, and $U: D \times L \rightarrow D$ is a total function called the update function.
\end{definition}

The update function $U$ can take different forms. A common scenario is when $L = D$, and then $U: D \times D \rightarrow D$ is a binary function. Another interesting case is when the update is \em parameter-independent \em, meaning the state transition is independent of the parameter $\ell \in L$. When the state space $D$ is finite, then $\sys$ can be thought of as a deterministic finite automaton (DFA), but without the initial state or accept states.

A {\em computation trace} in a semiautomaton with {\em initial state} $s_0$ is a (possibly infinite) sequence of states $s_0 \xrightarrow{} s_1 \xrightarrow{} s_2 \xrightarrow{} \dots$.
We say that $s \in D$ reaches a state $s' \in D$ if there is a finite sequence $s \rightarrow \dots \rightarrow s'$; we will use the shorthand notation $s \twoheadrightarrow s'$ for this and often say that $s$ reaches $s'$. We define $U^{k}(s)$ to be the set of states reachable from $s$ via a sequence of at most length $k$. We also define the {\em closure} of $s$, $U^\infty(s) = \bigcup_{k \geq 0} U^{k}(s)$ to be the set of states reachable from $s$.

\begin{definition}[Transition Graph]
Given a semiautomaton $\sys = (D,L,U)$, its transition graph $G[\sys]$ is a labelled directed graph with vertex set $D$ and edge set $\{(s,s'):\ell \mid s' =U(s,\ell)\}$.
\end{definition}

The output of a computation in a semiautomaton model is captured by a query $Q$, which is a total function $Q: D \rightarrow R$ that maps each state to an element of a (finite or infinite) query range $R$.
We can now define the main notion in this paper.

\begin{definition}[Free Termination State]
Given a semiautomaton $\sys = (D,L,U)$ and a query $Q: D \rightarrow R$, we say that a state $s \in D$ is a \textit{free termination state} if for all states $s' \in U^\infty(s): Q(s) = Q(s')$.
\end{definition}

In other words, if $s$ is a free termination state, any computation trace with initial state $s$ will leave the query result unchanged. This means that the distributed system can output the value of $Q$ without the need to continue the computation.

This model is inspired by algebraic models of data systems such as CRDTs and the group and ring models of incremental view maintenance~\cite{dbsp, rings-ivm-koch}. It is also inspired by application-level consistency~\cite{seeing-is-believing, consistency-without-borders} which considers only the observable state to users (our query result) as important for application guarantees rather than the internal system state. 

The behavior of free termination states can vary significantly. In the example below, we will present some common scenarios of free termination using DFAs; we will see some of these cases throughout the paper. 

\begin{example}

To map a DFA into the model of queries over semiautomatons, the transitions are the same in both models, but the query $Q$ for the semiautomaton returns true if the DFA state is accepting, otherwise false. Figure~\ref{fig:dfa} shows four different scenarios of free termination:
\begin{enumerate}
    \item A free termination state is reachable from each state and all free termination states return the same query result (see Figure~\ref{fig:cat1}). 
    \item All paths lead to a free termination state but different free termination states return different query results (see Figure~\ref{fig:cat2}). 
    \item There are no free termination states (see Figure~\ref{fig:cat3}). 
    \item Some paths lead to free termination states and others do not (see Figure~\ref{fig:cat4}). 
\end{enumerate}

We will see that certain algebraic properties of the semiautomaton and the query imply that we have to fall in one of these cases. For example, if the query is commutative (Appendix ~\ref{sec:commutative}), then free termination falls in the first scenario, while if the semiautomaton forms a group structure this means that no free termination state exists (Section ~\ref{sec:groups}).
\end{example}

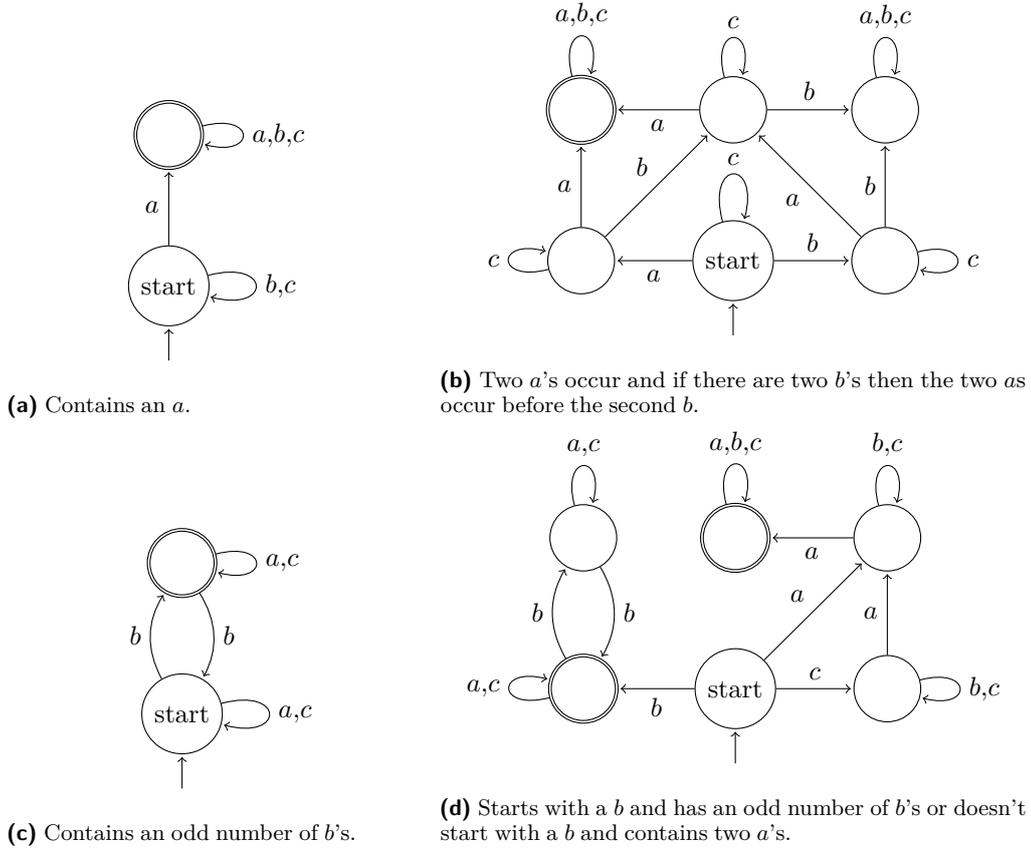
\begin{figure}
    \centering
\begin{subfigure}{0.4\textwidth}    
\centering
\begin{tikzpicture}[shorten >=1pt,node distance=2cm,on grid,auto,initial text=,initial below]
\node[state, initial] (s1) {start};
\node[state, accepting] (s2) [above of=s1] {};
\path[->]
  (s1) edge                node {$a$}  (s2)
  (s1) edge  [loop right]  node {$b,c$}  ()
  (s2) edge  [loop right]  node {$a,b,c$}  ();
\end{tikzpicture}
\caption{Contains an $a$.} \label{fig:cat1}
\end{subfigure}
\begin{subfigure}{0.55\textwidth}   
\centering
\begin{tikzpicture}[shorten >=1pt,node distance=2cm,on grid,auto,initial text=,initial below]
\node[state, initial] (s1) {start};
\node[state] (s2)  [above of=s1] {};
\node[state] (l1) [left of=s1] {};
\node[state,accepting] (l2) [above of=l1] {};
\node[state] (r1) [right of=s1] {};
\node[state] (r2) [above of=r1] {};
\path[->]
  (s1) edge  [loop above]  node {$c$}  ()
  (s1) edge  node {$a$}  (l1)
  (s1) edge  node {$b$}  (r1)
  (l1) edge  [loop left]  node {$c$}  ()
  (l1) edge  node {$a$}  (l2)
  (l1) edge  node {$b$}  (s2)
  (l2) edge  [loop above] node {$a,b,c$}  ()
  (r2) edge  [loop above] node {$a,b,c$}  ()
  (r1) edge  [loop right]  node {$c$}  ()
  (r1) edge  node {$a$}  (s2)
  (r1) edge  node {$b$}  (r2) 
  (s2) edge  [loop above]  node {$c$}  ()
  (s2) edge  node {$a$}  (l2)
  (s2) edge  node {$b$}  (r2) ;
\end{tikzpicture}
\caption{Two $a$'s occur and if there are two $b$'s then the two $a$s occur before the second $b$.} \label{fig:cat2}
\end{subfigure}
\begin{subfigure}{0.4\textwidth} 
\centering
\begin{tikzpicture}[shorten >=1pt,node distance=2cm,on grid,auto,initial text=,initial below]
\node[state, initial] (s1) {start};
\node[state, accepting] (s2) [above of=s1] {};
\path[->]
  (s1) edge  [bend left]  node {$b$}  (s2)
  (s2) edge  [bend left] node {$b$}  (s1)
  (s1) edge  [loop right]  node {$a,c$}  ()
  (s2) edge  [loop right]  node {$a,c$}  ();
\end{tikzpicture}
\caption{Contains an odd number of $b$'s.} \label{fig:cat3}

\end{subfigure}
\begin{subfigure}{0.55\textwidth} 
\centering
\begin{tikzpicture}[shorten >=1pt,node distance=2cm,on grid,auto,initial text=,initial below]
\node[state, initial] (s1) {start};
\node[state,accepting] (l1) [left of=s1] {};
\node[state] (l2) [above of=l1] {};
\node[state] (r1) [right of=s1] {};
\node[state] (r2) [above of=r1] {};
\node[state,accepting] (r3) [above of=s1] {};
\path[->]
  (s1) edge  node {$c$}  (r1)
  (s1) edge  node {$a$}  (r2)
  (s1) edge  node {$b$}  (l1)
  (l1) edge  [loop left]  node {$a,c$}  ()
  (l1) edge  [bend left] node {$b$}  (l2)
  (l2) edge  [bend left] node {$b$}  (l1)
  (l2) edge  [loop above] node {$a,c$}  ()
  (r1) edge  [loop right] node {$b,c$}  ()
  (r1) edge  node {$a$}  (r2)
  (r2) edge  [loop above] node {$b,c$}  ()
  (r2) edge  node {$a$}  (r3) 
  (r3) edge  [loop above]  node {$a,b,c$}  ();
\end{tikzpicture}
\caption{Starts with a $b$ and has an odd number of $b$'s or doesn't start with a $b$ and contains two $a$'s.} \label{fig:cat4}
\end{subfigure}
\caption{Four DFAs with labels $\{a,b,c\}$ showing four different categories of free termination. The doubly-lined states represent the accepting states of the DFA (query is true).} \label{fig:dfa}
\end{figure} 

\subsection{Applications of Free Termination}

We will next describe a few concrete instantiations of our computational model, which will be used throughout this paper.

\introparagraph{Incremental Query Computation}
In the first instantiation, we consider $D$ to be the set of all instances over a fixed relational schema $\mathbf{R}$. The update function $U$ allows modification only by inserting a tuple to the instance; in other words, the set $L$ is the set of all tuples over $\mathbf{R}$ and $U(I,t) = I \cup \{t\}$. We will refer to this semiautomaton as $\sys_\cup$.
A trace in $\sys_\cup$ is a sequence of single-tuple insertions to the initial database. Let us consider a query $Q$ that is a Boolean Conjunctive Query, hence it maps the database instance (state) to its range $R = \{\mathsf{T}, \mathsf{F}\}$.

The characterization of free termination states for $Q$ is fairly straightforward. Indeed, take any instance $I$ in the state space such that $Q(I)$ is true; then, because $Q$ is a monotone query and the updates are only tuple insertions, any $I'$ reachable from $I$ will also have $Q(I') = \mathsf{T}$. Thus, such an $I$ will be a free termination state. On the other hand, if $Q(I) = \mathsf{F}$, free termination is not possible since we can always insert a sequence of tuples to make $Q$ true. As we will see later on, this is a special case of a more general characterization of free termination.

\introparagraph{State-based CRDTs}
A state-based CRDT defines three functions over its state $D$: an \textit{update} operation that allows state to be mutated from outside the system ($\textsf{update}: D \times L \rightarrow D$), a \textit{merge} operation that determines how the states of two replicas can combine to converge to the same state ($\textsf{merge}: D \times D \rightarrow D$), and a \textit{query} operation that defines what is visible to the user from the internal state ($Q: D \rightarrow R$). The basic idea is that updates from users can modify the state of any replica in the distributed system and asynchronous gossip of replica states in the background will ensure each replica eventually converges to the same state regardless of duplicated or out-of-order message delivery. To guarantee this convergence (formally known as \textit{strong eventual consistency}), the merge operation must be associative, commutative, and idempotent and the update operation must be inflationary with respect to the partial ordering induced by the merge operation. In Appendix ~\ref{sec:crdt-examples}, we give four popular examples of state-based CRDTs.

To capture a state-based CRDT in our model, we define it as a semiautomaton $\sys_\square = (D,L,U)$, where $D$ is the set of states, and the update $U$ captures both the update and merge operation of the CRDT. Specifically, $U(s, \langle \textsf{u}, \ell \rangle) = \textsf{update}(s,\ell)$ and $U(s, \langle \textsf{m}, s' \rangle) = \textsf{merge}(s,s')$.
A crucial property of a state-based CRDT is that the partial order induced by the merge operation forms a join-semilattice. Observe that because of this, the state transition system for a state-based CRDT will always be acyclic - a property that we will explore in detail in Section ~\ref{sec:partial-orders}.

\introparagraph{Fixpoint Computation}
Consider a partial order $\sqsubseteq$ over a domain $D$, and take $f: D \rightarrow D$ to be a monotone function w.r.t. $\sqsubseteq$. Further consider a (parameter-independent) semiautomaton and assume that it captures a fixpoint computation from a starting state $s_0 \rightarrow f(s_0) \rightarrow f(f(s_0)) \rightarrow \dots$ that eventually reaches a state $s$ with $f(s)=s$ after finitely many steps. For any query $Q$ over $D$, the fixpoint state $s$ is a free termination state (no other state is reachable from it). However, the structure of $Q$ may allow us to find an earlier free termination state before we even reach the fixpoint.

As an example, consider the case where the fixpoint computation is an iterative (naive or semi-naive) evaluation of a Datalog program $P$ on an instance $I$. Here, the starting state is the set of \textsf{EDB}s and each iteration during evaluation is an update that adds to the current state the newly produced \textsf{IDB} facts via the rules of the program. Our query $Q$ is a view over the \textsf{IDB}s of the program $P$, which can be thought of as the "target" of $P$. Concretely, consider the following Datalog program that determines whether there exists a path between vertices $s$ and $t$ in a graph.
\begin{align*}
P(x,y) & \leftarrow Edge(x,y) \\
P(x,y) & \leftarrow P(x,z), Edge(z,y) \\
Q() & \leftarrow P(s,t)
\end{align*}

For this program, we can freely terminate as long as the state contains the \textsf{IDB} fact $P(s,t)$, even before we have reached a fixpoint. This is because the fixpoint computation is monotone w.r.t. set containment (so the update $U$ is inflationary) and $Q$ is a monotone query as well w.r.t. the partial order $\mathsf{F} \sqsubseteq \mathsf{T}$. If we replaced $Q$ with another query $Q'() \leftarrow P(x,x)$ which detects the presence of a cycle in the graph, we also can freely terminate earlier as long as the state contains a cycle. As we will see in the next section, both of these are examples of Boolean threshold queries.

Datalog$^o$~\cite{datalogo} is an extension of Datalog to support recursive queries over partially ordered pre-semirings (POPS). This can also be viewed as a parameter-independent fixpoint computation. Datalog$^o$ requires that both semiring operations be monotone w.r.t. the partial order of the POPS. This tells us that the update transition for any Datalog$^o$ graph is inflationary, which we show in Proposition ~\ref{prop:inflationary-threshold-ft} implies any Boolean monotone query will have a free termination state. For instance, the query $Q$ below that computes whether the distance between vertices $s,t$ in a graph is at most 10.
\begin{align*}
P(x,y) & \leftarrow min(Edge(x, y), min \{ (P(x,z) + Edge(z, y))) \\
Q() & \leftarrow P(s, t) \leq 10
\end{align*}

\section{Algebraic Properties and Free Termination}

In this section, we explore how the algebraic structures of the semiautomaton and the query affect free termination. 

\subsection{Partial Orders}
\label{sec:partial-orders}
Consider the example in Figure ~\ref{fig:threshold-query-example} of a semiautomaton in which update labels are singleton sets ($\{a\}$, $\{b\}$, or $\{c\}$) and the update applies set union of the incoming singleton to the current state. Labels on edges depict the incoming singleton that transitions from the source state to the destination state. We can think of this system as inputs from ($\{\{a\}$, $\{b\}$, or $\{c\}\}$) streaming in over time and we wish to compute some query over this stream of inputs without ever knowing if the stream has ended. The figure depicts the query "contains an $a$" over this semiautomaton, illustrated by the dotted green line that is the "threshold" after which states in the partial order return True in the query. Observe that the same graph with the edge labels and self-loops removed is the Hasse diagram of the set/subset partial ordering over atoms $\{a$, $b$,  $c\}$.

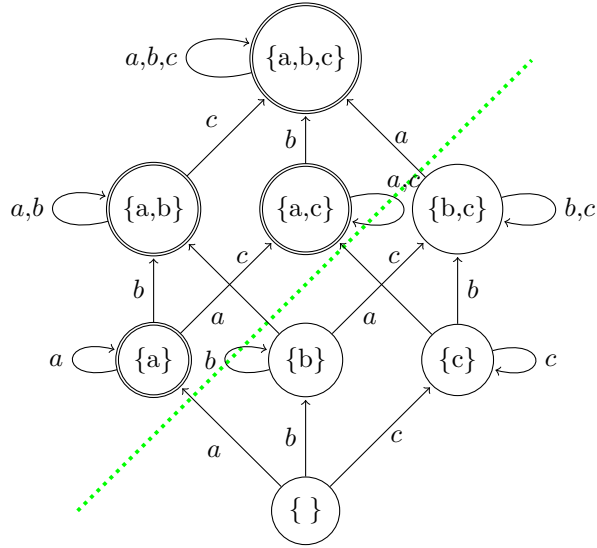
\begin{figure}
\centering
\begin{tikzpicture}[shorten >=1pt,node distance=2cm,on grid,auto,initial text=,initial below]
\draw[dotted, line width=0.5mm, color=green] (-3,0) -- (3,6);
\node[state] (bot) {\{ \}};
\node[state] (b) [above of=s1] {\{b\}};
\node[state, accepting] (a) [left of=b] {\{a\}};
\node[state] (c) [right of=b] {\{c\}};
\node[state,accepting] (ab) [above of=a] {\{a,b\}};
\node[state,accepting] (ac) [above of=b] {\{a,c\}};
\node[state] (bc) [above of=c] {\{b,c\}};
\node[state, accepting] (abc) [above of=ac] {\{a,b,c\}};
\path[->]
  (bot) edge  node {$b$}  (b)
  (bot) edge  node {$a$}  (a)
  (bot) edge  node[right] {$c$}  (c)
  (a) edge  node {$b$}  (ab)
  (a) edge  node[above, yshift=2mm, xshift=2mm] {$c$}  (ac)
  (b) edge  node[above, yshift=2mm, xshift=2mm] {$c$}  (bc)
  (b) edge  node[below, yshift=-2mm, xshift=-2mm] {$a$}  (ab)
  (c) edge  node[below, yshift=-2mm, xshift=-2mm] {$a$}  (ac)
  (c) edge  node[right] {$b$}  (bc)
  (ac) edge  node {$b$}  (abc)
  (ab) edge  node {$c$}  (abc)
  (bc) edge  node[right] {$a$}  (abc)
  (a) edge  [loop left]  node {$a$}  ()
  (b) edge  [loop left]  node {$b$}  ()
  (c) edge  [loop right]  node {$c$}  ()
  (ab) edge  [loop left]  node {$a,b$}  ()
  (bc) edge  [loop right]  node {$b,c$}  ()
  (ac) edge  [loop right]  node[above, yshift=1mm] {$a,c$}  ()
  (abc) edge  [loop left]  node {$a,b,c$}  ();
\end{tikzpicture}
\caption{A state transition system for the set union state transition over the universe $\{a,b,c\}$. The green dotted line across the states indicates the threshold line for the query "the set contains an $a$". The doubly-lined states return True and single-lined states return False in the query.} \label{fig:threshold-query-example}
\end{figure}

We will utilize properties of this example to reason about its free termination states. In this particular case, the update function is \textit{inflationary}. If we view the ordering of $\mathbb{B}$ as $\mathsf{F} \sqsubseteq \mathsf{T}$ then we can see that $Q$ is \textit{monotone}. Formally, to define what it means for $Q$ to be monotone we must have some partial order on the sets $D$ and $R$. Recall that a binary relation $\sqsubseteq$ is a partial order if it is reflexive, transitive and antisymmetric.

\begin{definition}[Inflationary]
Let $\sys = (D,L,U)$ be a semiautomaton equipped with a partial order $\sqsubseteq_D$ on $D$. $\sys$ is {\em inflationary (resp. deflationary)} w.r.t $\sqsubseteq_D$ if whenever $s \xrightarrow{} s'$ then  $ s \sqsubseteq_D s'$ (resp. $ s \sqsupseteq_D s'$).
\end{definition}

In an inflationary semiautomaton, any update of a state will always follow the underlying partial order of the state space. A deflationary system will follow the partial order in reverse. 

In many cases, we can define a "natural" partial order on $D$ via the update function $U$. We say that $U$ is {\em acyclic} in $\sys$ if the transition graph G[$\sys$] is acyclic (excluding self-loops). The following proposition follows from the fact that the transitive closure of a directed acyclic graph is a partial order.

\begin{proposition}
Let $\sys = (D,L,U)$ and $U$ be acyclic. Then, the relation $s \sqsubseteq_U s' \Leftrightarrow s \twoheadrightarrow s'$ is a partial order for $D$. Moreover, $\sys$ is inflationary w.r.t. $\sqsubseteq_U$.
\end{proposition}

Whenever $U$ is acyclic, we will refer to the partial order $\sqsubseteq_U$ as the {\em natural partial order} of $\sys$. This is akin to the notion of natural orders for algebraic structures such as pre-semirings~\cite{datalogo}. If the update transition is a binary operation $U:D \times D \rightarrow D$ forming a monoid (is associative and has an identity) then the natural ordering would be the standard $x \sqsubseteq_U y$ iff $\exists z : x \xrightarrow{z} y$.

\begin{example}
Consider again the incremental query computation semiautomaton $\sys_\cup$. The update operation is acyclic (since an update only adds tuples to the instance), and the natural partial order is set containment of the states.
\end{example}

\begin{definition}[Monotone Query]
Let $\sys$ be a semiautomaton equipped with a partial order $\sqsubseteq_D$ on $D$. Let $R$ be a set with partial order $\sqsubseteq_R$. A  query $Q: D \rightarrow R$ is {\em monotone (resp. antitone)} w.r.t. $\sys$ if $s \sqsubseteq_D s'$ implies $Q(s) \sqsubseteq_R Q(s')$ (resp. $Q(s) \sqsupseteq_R Q(s')$).
\end{definition}

\introparagraph{Free Termination Conditions}
Intuitively, if the input state will only increase in the partial order because $U$ is inflationary, and $Q$ is monotone, then over time the output of $Q$ will always stay the same or increase. However, free termination is concerned specifically with when the output of $Q$ will {\em stay the same}. We identify two general conditions that guarantee free termination in this case.

\begin{proposition}
\label{inflationary-monotone-implies-ft}
Let $\sys = (D,L,U)$ be a semiautomaton equipped with a partial order $\sqsubseteq_D$. If $\sys$ is inflationary (resp. deflationary) and $s \in D$ is a maximal (resp. minimal) element of $\sqsubseteq_D$, then $s$ is a free termination state.
\end{proposition}

\begin{proof}
Consider any state $s' \in D$ such that $s' \in U^\infty(s)$. Then, because $\sys$ is inflationary and $\sqsubseteq_D$ is transitive, we must have that $s \sqsubseteq_D s'$. But $s$ is a maximal element of  $\sqsubseteq_D$, hence $s = s'$. Thus, $Q(s) = Q(s')$. The proof for when $U$ is deflationary and $s$ a minimal element is symmetrical.
\end{proof}

In some cases, we will consider partial orders $\sqsubseteq_D$ with a bottom element $\bot$ or a top element $\top$: then, we have that for every $s \in D$, $\bot \sqsubseteq_D s$ or $s \sqsubseteq_D \top$ respectively. Note that if a top element exists, then $\top$ is a free termination state by the above proposition.


\begin{proposition}
Let $\sys = (D,L,U)$ be a semiautomaton with a  partial order $\sqsubseteq_D$. Let $Q: D \rightarrow R$ be a query, with $R$ equipped with a partial order $\sqsubseteq_R$. If $\sys$ is inflationary, $Q$ is monotone (resp. antitone), and $Q(s)$ is a maximal (resp. minimal) element of $\sqsubseteq_R$,  then $s \in D$ is a free termination state.
\label{top-in-r-free-terminates}
\end{proposition}

\begin{proof} 
Consider any state $s' \in D$ such that $s' \in U^\infty(s)$. Then, because $\sys$ is inflationary and $\sqsubseteq_D$ is transitive, we must have that $s \sqsubseteq_D s'$. From the monotonicity of $Q$, this implies that $Q(s) \sqsubseteq_R Q(s')$. But because $Q(s)$ is maximal, we must have that $Q(s') = Q(s)$. The case where $Q$ is antitone and $Q(s)$ is a minimal element is symmetric.
\end{proof}




\begin{example}
A classic example of an antitone query over $\mathbb{B}$ is the use of a universal quantifier $\forall$. Consider the natural inverse of the existential quantifier $\exists$ used in our first example. If we let $Q$ be “every element of the input stream is an $\{a\}$” then we can freely terminate when the answer becomes $\mathsf{F}$ (as soon as any non-"a" streams in).
\end{example}

\introparagraph{Threshold Queries}
In this part, we ask the following question: what classes of queries have free termination states? We will restrict our attention to settings where the behavior of $\sys$ is inflationary as threshold queries are naturally valuable in this setting. Further, all state-based CRDTs are examples of $\sys$ that are inflationary. We define first an important class of queries, inspired by the LVars~\cite{lvars} work on coordination-free programming languages, called Threshold Queries. 
Recall that an {\em antichain} in a partial order $\sqsubseteq$ with domain $D$ is a subset $C \subseteq D$ such that no two distinct elements of $C$ are comparable under $\sqsubseteq$.

\begin{definition}[Boolean Threshold Query]
Let $\sqsubseteq$ be a partial order with domain $D$, and $C \subseteq D$ be an antichain. 
A \textit{Boolean threshold query} $Q_C: D \rightarrow \mathbb{B}$ with threshold line $C$ is the Boolean query that returns the value $Q_C(s) = \bigvee_{c \in C} (s \sqsupseteq c)$. 
\end{definition}

Boolean threshold queries are monotone and this allows us to obtain the free termination states in this case.

\label{monotone-boolean-implies-threshold}

\begin{proposition}
\label{prop:inflationary-threshold-ft}
Let $\sys = (D,L,U)$ be an inflationary semiautomaton equipped with a partial order $\sqsubseteq_D$. Let $Q_C$ be a Boolean threshold query. Then, the free termination states  are exactly the elements of $D$ that are at or above $C$.
\end{proposition}

\begin{proof}
We first show that $Q_C$ must be monotone. Indeed, let $s \sqsubseteq_D s'$ and assume $Q_C(s)$ is true. Then, there exists $c \in C$ such that  $s \sqsupseteq_D c$. But since $s \sqsubseteq_D s'$, $s' \sqsupseteq_D c$ and $Q_C(s')$ must also be true.
 Since $\mathsf{T}$ is the maximal element of $\mathbb{B}$, all states $s$ such that $Q_C(s)$ is true are free termination states. Note that any element of $D$ that is at or above the threshold line $C$ has this property.
\end{proof}

We should note here that any monotone Boolean query must be a Boolean threshold query (excluding the trivial query that is always false).
In this work, we do not wish to limit ourselves to threshold queries that return only true, or even to thresholds where each $c$ returns the same value $Q(c)$. We can generalize the concept of Boolean threshold queries with the following result.

\begin{proposition} \label{acyclic-implies-threshold}
Let $Q$ be a query with at least one free termination state for a semiautomaton $\sys = (D, L, U)$ where $U$ is acyclic. Then, there exists an antichain $C \subseteq D$ w.r.t. the natural partial order $\sqsubseteq_U$ such that whenever $c \sqsubseteq_U s$ for some $c \in C$, $Q(s) = Q(c)$.
\end{proposition}

\begin{proof}
Let $F \subseteq D$ be the set of free termination states. $F$ cannot be empty. Let $F_{\tau}$ be the minimal states in $F$ w.r.t. $\sqsubseteq_U$. By construction, $F_\tau$ forms an antichain and will be our set $C$. Take now any state $s \in D$ such that $c \sqsubseteq_U s$ for some $c \in C$. Since $c$ is a free termination state and $c$ reaches $s$, it must be that $Q(c) = Q(s)$. 
\end{proof}

In other words, there exists a threshold (formed by the antichain $C$) such that at or above the threshold the behavior of $Q$ is governed completely by the threshold states. However, the behavior of $Q$ outside of this threshold space has no restriction (unlike in the case of Boolean monotone queries). 

\introparagraph{Join Semilattices} When the partial order has further algebraic structure, we can characterize the behavior of free termination states more precisely. Here, we consider the case where $\sys$ has a natural partial order $\sqsubseteq_U$ and this order is a join-semilattice (this means that there is a least upper bound for any nonempty finite subset of $D$). 

\begin{proposition} \label{prop:semilattice}
Let $\sys = (D,L,U)$ be a semiautomaton where $\sqsubseteq_U$ forms a join-semilattice. Then, for any query $Q$, all free termination states return the same value.
\end{proposition}

\begin{proof}
Indeed, take any two free termination states $s_1,s_2$. Since $\sqsubseteq_U$ is a join-semilattice, there exists a least upper bound $s$ such that $s_1 \sqsubseteq_U s$ and $s_2 \sqsubseteq_U s$. Since $\sqsubseteq_U$ is the natural partial order, $s_1 \twoheadrightarrow  s$ and $s_2 \twoheadrightarrow  s$. But $s_1,s_2$ are free termination states, so it must be that $Q(s_1) = Q(s)$ and $Q(s_2) = Q(s)$. Thus, $Q(s_1) = Q(s_2)$.
\end{proof}

As a corollary of the above proposition, any free termination state in a state-based CRDT must have the same value. A special case is when the natural partial order has a top element $\top$; in this case, all free termination states must take the value of $Q(\top)$.

\begin{proposition}
\label{prop:fts-reachable-in-semilattice}
Let $\sys = (D,L,U)$ be a semiautomaton where $\sqsubseteq_U$ forms a join-semilattice. If a query $Q$ has a free termination state in $\sys$, then any state can reach a free termination state.
\end{proposition}

In other words, if there exists a free termination state, it is possible to reach a free termination state from whichever state we currently are in (with an appropriate set of updates).

\begin{proof}
Let $s$ be the current state and $s_t$ be a free termination state. Since $\sqsubseteq_U$ is a join-semilattice, there exists a least upper bound $s'$ such that $s \sqsubseteq_U s'$ and $s_t \sqsubseteq_U s'$. Since $\sqsubseteq_U$ is the natural partial order, $s \rightarrow s'$ and $s_t \rightarrow s'$. But then $s'$ must also be a free termination state.
\end{proof}

One would be tempted to think that monotone (or antitone) queries are the only ones that have free termination states in a join-semilattice, but this is not true even for Boolean queries. Take for example the semiautomaton $\sys_\cup$ with the Boolean query $Q() = R(c) \wedge \neg S(c)$ for some constant $c$. This query is neither monotone nor antitone, since, for example, it returns false on $\{R(a)\}$, true on $\{R(a), R(c)\}$, and false on $\{R(a), R(c), S(c)\}$. However, it has free termination states: these are the states in which the tuple $S(c)$ is in the instance.

\subsection{Group-Like Structures}
\label{sec:groups}
If $U$ is not acyclic, then we can construct examples where free termination is not possible. We will use the notions of inverses and identity values from groups, but to generalize beyond binary update operations $U: D \times D \rightarrow D$ we must define these terms for the general case of transition graphs.

\begin{definition}
Let $\sys = (D,L,U)$ be a semiautomaton. A state $\mathsf{id} \in D$ is called an {\em identity state} if $\mathsf{id} \twoheadrightarrow s$ for every $s \in D$. A state $s \in D$ is called {\em invertible} if it can reach an identity state. 
\end{definition}

We call the following theorem "the inverse curse theorem".

\begin{theorem}
Let $\sys = (D,L,U)$ be a semiautomaton, and $Q: D \rightarrow R$ be a non-constant query. If every state of $D$ is invertible, then $Q$ has no free termination states in $\sys$.
\end{theorem}

\begin{proof}
Suppose that some state $s \in D$ freely terminates. Since $Q$ is not constant, there exists some state $s' \in D$ such that $Q(s) \neq Q(s')$. Since $s$ is invertible, we have that $s  \twoheadrightarrow \mathsf{id} $, where $\mathsf{id}$ is an identity element. By the definition of an identity element, $\mathsf{id} \twoheadrightarrow s'$. Hence, $s \twoheadrightarrow s'$, which contradicts the fact that $s$ is a free termination state.
\end{proof}

\begin{example}
As an example, consider the case where $D = \mathbb{Z}$, and let $U(i,+) = i+1$ and $U(i,-) = i-1$, i.e., we have a counter that can be incremented or decremented. Then, any state in $D$ is an identity state and invertible, so any query $Q$ has no free termination states unless it is constant.
\end{example}

When $U: D \times D \rightarrow D$ and $(D,U)$ forms a group ($U$ has identity element and every element has an inverse), then we obtain the following corollary:

\begin{corollary}
Let $\sys = (D,L,U)$ be a semiautomaton such that $(D,U)$ forms a group. If $Q$ is not a constant query, then $Q$ has no free termination states.
\end{corollary}

This corollary tells us that an update function that forms a group precludes the possibility of free termination. It also tells us that in view maintenance, which often studies rings rather than semirings, free termination is impossible. This is consistent with the monotonicity lens on threshold queries from CRDTs~\cite{crdtOn} where an inverse corresponds to moving backward in the partial order, which prevents the possibility of threshold queries with free termination states. The practical benefits of having inverses in CRDTs to allow "undo" operations have been discussed in ~\cite{undoableCrdts} and ~\cite{latticeWrappers}.
This is an interesting dichotomy. Two parallel lines of work have shown the value of invertibility in data systems (DBSP~\cite{dbsp}, DBToaster~\cite{dbtoaster}) and the value of coordination-free monotone queries (CALM theorem, CRDTs), but the benefits of these properties appear mutually exclusive.

\section{Free Termination in Distributed Systems}
\label{sec:distributed-systems}

In this section, we will study the problem of distributively computing a query in a coordination-free manner via the lens of free termination. Both transducer networks~\cite{calmTheorem} and CRDTs offer coordination-free models of eventual convergence, i.e., the existence of a time (called the quiescence point) where the query result converges.  However,  without the ability of a machine to determine whether it has already reached such a quiescence point, the applicability of this notion in practical systems is limited~\cite{crdtOn}.  If a user wants to read the (complete) output of a transducer network, eventual consistency cannot provide a certain answer.

By definition, free termination can determine when a complete answer can be given with certainty and hence aligns with the need for distributed systems to promptly respond to user requests. However, this stronger guarantee warrants a more fine-grained definition of coordination-freeness. In particular, we will need to talk about coordination-freeness as a property of a query and its input. This is in contrast to the CALM Theorem~\cite{calmTheorem}, which talks about coordination-freeness as a property of an entire query. 

\subsection{Distributed Model}

Our task is to compute a query $Q$ over a relational instance $I$ defined over a schema $\mathbf{R}$ over a network $\mathcal{N}$. A network $\mathcal{N}$ is defined as a finite, connected, undirected graph over a set of vertices $V$. Initially, the instance $I$ is horizontally partitioned across the nodes in the network.

The distributed computational model we consider here is a simplified version of the relational transducer network model used by Ameloot et al. In particular, we leverage the equivalence between oblivious transducers and coordination-freeness proven in~\cite{calmTheorem} to focus on oblivious networks and specifically a type of communication protocol used by many constructions in ~\cite{calmTheorem} called \em network flooding \em. 
In this construction, nodes attempt to achieve eventual consistency by broadcasting all their local information to their neighbors and sending no other messages.
More precisely, the computation in each node will be captured by the semiautomaton $\sys_\cup$, where each state is an instance over $\mathbf{R}$ and each transition adds a new tuple in the instance. This roughly corresponds to an oblivious, inflationary, and monotone relational transducer, with the critical difference that there is no output generated -- only the state is modified. 
Nodes can communicate by sending a tuple from their local state to be added to the instance of any neighboring node.

For now, we will focus on Boolean queries. Formally, a {\em configuration} of the network is a triple of mappings $\gamma = (state, buf, ready)$, where $state$ maps each node in $V$ to a state in $\sys_\cup$, $buf$ maps each node in $V$ to a finite multiset of facts from $\mathbf{R}$ called a buffer, and $ready$ maps each node to $\{\mathsf{F}, \mathsf{T}\}$. Initially, each state contains only the tuples in $I$ from the local partition, the buffers are empty, and $ready(v) = \mathsf{F}$. There are three types of transitions between two configurations:
\begin{description}
    \item[Produce Transition:] A node can move any tuple in its local state to the buffers of its neighboring nodes;
    \item[Consume Transition:] A node can update its state by removing a tuple from its local buffer and adding it to its local instance.
    \item[Ready Transition:] A node can set $ready(v) \gets \mathsf{T}$.
\end{description}

Importantly, once $ready$ is set to true, it cannot be further modified. Setting $ready$ to true denotes that the query result will not change. An algorithm for $Q$ in this model determines two things: $(i)$ when to send each tuple to its neighbors, and $(ii)$ if and when to perform a ready transition.

A {\em run} $\rho$ is an infinite sequence of transitions starting from an initial configuration. A run is {\em fair} if every fact in every buffer is eventually taken out, and it is {\em complete} if every tuple in a state is eventually sent to its neighbors. Finally, we say that an algorithm computes the correct output of $Q$ if for all inputs $I$ and all horizontal partitionings of $I$, whenever $ready(v) = \mathsf{T}$, $Q(state(v)) = Q(I)$. 
In every infinite run, there is a natural number $n>0$ such that none of the states change after the $n$-th transition: we call this the {\em quiescence point}. In a fair and complete run, each local state eventually converges to $I$, and hence at the quiescence point $state(v) = I$ for every node $v \in V$. However, without coordination, it is not possible to know when a node has received the entire input. 

\begin{definition}[Coordination-Free Correctness]
We say that the pair $(Q,I)$ is {\em coordination-free correct} if there exists an algorithm such that in every fair and complete run the correct output of $Q$ is computed and $ready$ is set to true at all nodes.
\end{definition}

In other words, no matter how computation proceeds and how messages are exchanged, at some point, the algorithm will perform the ready transition and hence know (without coordination) that it has computed the correct result and thus can give it to the user.
The next proposition relates the above notion of fine-grained coordination-freeness to free termination.

\begin{theorem} \label{thm:coord:free}
$(Q,I)$ is coordination-free correct if and only if $I$ is a free termination state for $Q$ in the semiautomaton $\sys_\cup$.
\end{theorem}

\begin{proof} 
Suppose $I$ is a free termination state. Consider the following algorithm: it will set $ready(v)$ to true exactly when $state(v)$ is a free termination state for $Q$ in $\sys_\cup$. This algorithm is correct, since at the instance $I'$ where $ready$ becomes true, all reachable states maintain the result of $Q$. Since $I$ is reachable from $I'$ (from the network flooding construction), $Q(I) = Q(I')$. Further, every fair and complete run reaches the quiescence point, when $state(v) = I$, and thus the algorithm will set $ready$ to true.

In the other direction, suppose $(Q,I)$ is coordination-free correct and consider an algorithm that computes the output correctly. At the quiescence point, when $state(v) = I$, this algorithm must set $ready$ to true (since the state will remain unchanged from that point on). But now consider a fair and complete run for another input $I' \supseteq I$ (which is a reachable state in $\sys_\cup$) such that some node $v$ receives first all of $I$. At this point, the algorithm would need to do the ready transition. But because of correctness, it must be that $Q(I) = Q(I')$. Hence, $I$ is a free termination state.
\end{proof}

If the input is not a free termination state, the system will often need to perform some coordination to get the user a concrete answer. We thus avoid coordination for a given query \textit{on some inputs, but not all inputs!}  In practice, the distribution of inputs to a given system is what determines how helpful free termination is. 

We next discuss how we could define a notion of coordination-free correctness for the entire query $Q$. As a first attempt, we could define $Q$ to be coordination-free correct if $(Q,I)$ is coordination-free correct for every input $I$. From Theorem~\ref{thm:coord:free}, this is equivalent to saying that every input is free terminating, which happens only when $Q$ is a constant query. Hence, we need to slightly relax this notion, by requiring that $(Q,I)$ is coordination-free correct for some inputs.

\begin{definition}
 We say that a Boolean query $Q$ is {\em positively (resp. negatively) coordination-free} if $(Q,I)$ is coordination-free correct for every input $I$ such that $Q(I) = \mathsf{T}$ (resp. $Q(I) = \mathsf{F}$).   
\end{definition}

The notion of positive coordination-freeness is exactly the notion of query coordination-freeness used for transducer networks in~\cite{calmTheorem}. Indeed, a transducer can only write true in its output tape, so if nothing is written, we assume a false output. 

\begin{theorem}
A Boolean query $Q$ is positively (resp. negatively) coordination-free if and only if $Q$ is monotone (resp. antitone).
\end{theorem}

\begin{proof}
From Theorem~\ref{thm:coord:free}, $Q$ is positively coordination-free if and only if every state $I$ with $Q(I) = \mathsf{T}$ is a free termination state. Say $I \subseteq I'$ and suppose $Q(I)$ is true. Then, $I$ is a free termination state and thus $I'$ is as well, which implies $Q(I') = Q(I) = \mathsf{T}$. Hence, $Q$ is monotone. The antitone case is symmetric.
\end{proof}

For instance, consider the Boolean query $\forall x: R(x) > 0$. This query freely terminates on all false instances, and hence it can be considered negatively coordination-free. Observe that Ameloot et al.~\cite{calmTheorem} categorizes antitone Boolean queries as not being computable by oblivious transducer networks. This is because of the definition of the output of a transducer network being the union of outputs and the encoding of the boolean values True and False being the presence of an empty tuple and the absence of a tuple respectively. In addition to antitone queries, some queries are neither monotone nor antitone, but still have coordination-free correct inputs such as our example from Section ~\ref{sec:partial-orders}, $Q() = R(c) \wedge \neg S(c)$.

\introparagraph{Non-Boolean Queries} Consider now a non-Boolean query $Q$ that outputs a set. In this case, we will introduce a ready variable $ready(v,t)$ for every node $v$ and every potential tuple $t$. Note that this introduces another layer of granularity since we can now compute some tuples in a coordination-free correct manner, while others we cannot. We can similarly lift coordination-free correctness to the entire query $Q$ by saying that $Q$ is positively (resp. negatively) coordination-free if every Boolean query $(t \in Q(I))$ is positively (resp. negatively) coordination-free. Positive coordination-freeness captures exactly how oblivious relational transducers work: they can write only correct facts to the output tape, which they can never retract. The following result is immediate.

\begin{theorem}
A non-Boolean query $Q$ is positively (resp. negatively) coordination-free if and only if $Q$ is monotone (resp. antitone).
\end{theorem}

We complete this section by an example that shows the benefits of a fine-grained coordination-free definition. Consider the following query $Q$, which is neither monotone nor antitone: $Q(x) = (R(x) \wedge x > 10) \vee (S(x) \wedge \neg T(x))$. In this case, any instance that contains a tuple $R(20)$ freely terminates (with true) for the output tuple $(20)$. Also, any instance that contains $(T(5))$ freely terminates (with false) for the output tuple $(5)$.  Thus, we can correctly output the existence/non-existence of the two tuples in a coordination-free manner.

\subsection{Distributed Computation with Metadata}

Many distributed systems depend on some form of metadata to infer facts about what state transitions may be possible in the future. The essential such example from ~\cite{calmTheorem} is the \textsf{all()} relation which returns the IDs of all transducers in a network. \textsf{all()} is used to compute non-monotone queries and can also be modeled with free termination.

To encode the \textsf{all()} relation, we extend the schema $\mathbf{R}$ to include a nullary relation $All()$ that will be set to true if we know that all machines have sent all their local data. Formally, we will extend $\sys_\cup$ such that the update transitions of any state with $All = \mathsf{F}$ will be as before, but all outgoing transitions of a state with $All = \mathsf{T}$ will be self-loops. This makes every state with $All = \mathsf{T}$ be a free termination state. Since in every fair and complete run a node will receive all input data, we are always guaranteed that we will reach a free termination state and hence can do the ready transition correctly. Of course, the tradeoff is that updating $All$ requires coordination between the nodes.

In Appendix ~\ref{sec:appendix-metadata} we show how partitioning policy metadata in transducer networks~\cite{zinn-win-move, weakerMonotonicity} can be modeled with free termination as well.

\section{Free Termination with Finite States}
\label{sec:finite}

For practical applications, we are interested in computing free termination states both statically at compile time and dynamically at runtime. Thus far, we have identified algebraic properties of programs that allow us to detect free termination states. In this section, we study free termination when the semiautomaton has a finite state space.

\subsection{Detecting Free Termination}

At runtime, it is desirable to determine whether the current state is a free termination state. This can be done by computing the reachable states from the current state in the transition graph $G[\sys]$ and verifying that they all return the same query result. Hence, this computation is linear to the size of the transition graph. However, it turns out that we can compute all free terminating states in linear time as well.

\begin{proposition} \label{prop:free:linear}
Let $\sys=(D,L,U)$ be a semiautomaton with a finite state space, and $Q$ be a query. We can determine all free terminating states in time linear to the size of the transition graph (assuming computing $Q$ takes constant time). 
\end{proposition}

\begin{proof}
To achieve this result, we first need the following observation on the behavior of free terminating states: if $s$ is a free terminating state, then all states in the Strongly Connected Component (SCC) of $s$ in $G[\sys]$ are also free terminating. Thus if two states $s,s'$ are in the same SCC and $Q(s) \neq Q(s')$, then none of the states in the SCC are free terminating. 

The first step of the algorithm is to convert $G[\sys]$ to a Directed Acyclic Graph (DAG) $G'$, where each node represents an SCC. This is standard and can be done in linear time. As a second step, we iterate over all SCCs and label them as a candidate if $Q$ is the same across all states in that SCC; otherwise, we remove from $G'$ the SCC and all other nodes that can reach it. This step can also be implemented in linear time.

We are now left with a DAG $G''$ where each SCC has the same value for $Q$. In our final step, we perform a traversal of the SCCs in reverse topological order. Any SCC with no outgoing edge in $G''$ is free terminating (meaning all the states in the SCC are free terminating). At any step, if an SCC $C$ has outgoing edges only to free terminating SCCs and $Q$ is the same for $C$ and its outgoing SCCs, we mark $C$ as free terminating; otherwise, it is not free terminating. This final step also requires linear time.
\end{proof}

Recall that a finite semiautomaton $\sys$ is a DFA without a start or accept states. If we denote a state of $\sys$ to be the start state, and take $Q$ to be a Boolean query that returns true for accept states and false for the other states, we have exactly a DFA. In this case, free termination of a state means essentially that we can stop the computation of the DFA without the need to read any more symbols from the input. From Proposition~\ref{prop:free:linear}, we obtain as a corollary:

\begin{corollary}The free termination states of a deterministic finite automata (DFA) can be computed in linear time in the size of the DFA.
\end{corollary}

DFAs are an interesting example of converting an infinite state space (all strings over the alphabet) to a finite state space. We now explore connections between free termination and this notion of equivalent state representations.

\subsection{State Minimization}

\label{sec:state-minimization}

In this section, we ask whether given a semiautomaton $\sys$ and a query $Q$ we can construct another "simpler" semiautomaton $\sys'$ that has the same behavior as $\sys$ for the given query. Formally:

\begin{definition}[Equivalence]
 Let $\sys, \sys'$ be two semiautomata that both contain a start state $\bot$ that reaches all states and have the same label set $L$. Let $Q,Q'$ be two queries on $\sys, \sys'$ respectively.  We say that $(\sys,Q)$ is {\em equivalent} to $(\sys',Q')$ if given the same sequence of transitions starting from $\bot$, the query result will be identical.
\end{definition}

If $(\sys, Q)$ corresponds to a DFA, the above definition captures exactly DFA equivalence. In this case, one can simply perform state minimization in DFAs~\cite{sipserTextbook} to obtain a minimal DFA. The following property holds for free termination states in DFAs.

\begin{proposition}
A free termination state cannot be part of a cycle in a minimal DFA.
\end{proposition}

\begin{proof}
Assume for the sake of contradiction that we have a minimal DFA and a free termination state $s$ that participates in a cycle in that DFA. Each state in the cycle is reachable from $s$, so they must all return the same query result as $s$ (accept/reject) and also be free-termination states. The DFA in which each state in the cycle is collapsed into a single state is equivalent and has strictly fewer states, thus the original DFA cannot have been minimal.    
\end{proof}

For a general pair $(\sys,Q)$, we can follow the same idea as state minimization in DFAs. We define the \textit{collapsing} of a set of states $S \subseteq D$ that take the same value for $Q$ as the modification of the semiautomaton that $(i)$ replaces $D$ with $(D \setminus S) \cup \{s_0\}$, where $s_0$ is a new state, and $(ii)$ any update that transitions to a state in $S$ goes to $s_0$, and any update that transitions from a state in $S$ starts in $s_0$. We also set $Q(s_0)$ to be the value of a state in $S$.  

\begin{proposition} \label{prop:collapse}
Let $\sys = (D,L,U)$ be a semiautomaton and $Q$ be a query. Let $s \in D$ be a free terminating state. Then, the pair $(\sys',Q')$ resulting from collapsing $U^\infty(s)$ is equivalent to $(\sys,Q)$.
\end{proposition}

\begin{proof}
Because $s$ is a free termination state, we know every state in $U^\infty(s)$ returns the same query result as $Q(s)$. We can collapse all states in $U^\infty(s)$ with all transitions being self-loops. The collapsed state $s_0$ is a free termination state and cannot be part of a cycle. Any sequence of updates that ends up in $U^\infty(s)$ in $\sys$ will end up in $s_0$ in $\sys$, and hence we have the desired equivalence.
\end{proof}

Given $(\sys, Q)$ with finite state space, consider now the equivalent pair $(\sys_\square,Q_\square)$ obtained by repeatedly applying the above proposition to free termination states until there is no more change; we call $(\sys_\square,Q_\square)$ collapsed. This is analogous to a minimal state representation. Note that we need finiteness to guarantee that the process of collapsing states will terminate at some point. We can now show the following characterization of free termination in this case.

\begin{proposition}
Let $\sys$ be a semiautomaton with an initial state $\bot$ and finite state space, and $Q$ be a query. Then, for the equivalent collapsed pair $(\sys_\square,Q_\square)$, a state is a free termination state if and only if all outgoing transitions are self-loops (i.e. the state is a fixpoint).
\end{proposition}

\begin{proof}
 From Proposition \ref{prop:collapse}, we know that a free termination state that has any reachable states other than itself can be collapsed into a single state in which there are no outgoing edges except self-loops. The other direction is also clear as a state $s$ that has no outgoing transitions satisfies the definition of a free termination state as $U^{\infty}(s) = \{s\}$.   
\end{proof}

State minimization offers an interesting perspective on distributed systems techniques like CRDTs. While CRDTs appear to only describe inflationary state mutations, it is common to convert non-inflationary updates into inflationary ones using metadata. We see this in both the two-phase set CRDT and the positive-negative counter CRDT (appendix ~\ref{sec:crdt-examples}). While the state transition graphs for these CRDTs are acyclic, they have equivalent representations that are cyclic. By looking at these structures from the perspective of user-visible state (query-layer equivalence), the separation between free termination of query results and eventual consistency of state becomes clear.

\section{Conclusion}
We have presented a general state transition framework for reasoning about coordination-free computation in distributed systems. Our central notion, \textit{free termination}, allows the relational transducer approach for declarative networking and the algebraic lattices approach of conflict-free replicated data types to be modeled in a single framework.

\section{Acknowledgments}
We thank Laura Power, Val Tannen, and Dan Suciu for their insightful feedback on this work. This work was made possible by the Simons Institute 2023 session on Logic and Algorithms in Database Theory and AI. This work is supported in part by gifts from Accenture, AMD, Anyscale, Cisco, Google, IBM, Intel, Intesa Sanpaolo, Lambda, Microsoft, NVIDIA, Samsung SDS, SAP, and VMware.

\bibliographystyle{plainurl}
\bibliography{bibliography}

\newpage
\appendix
\section{Commutativity}
\label{sec:commutative}
Another algebraic property that many distributed systems satisfy is that of commutativity. It will be convenient to switch from polish prefix notation to infix notation for our update transitions. We use $\cdot$ for the application of an update transition so $U(s, a)$ becomes $s \cdot a$ and $U(U(s, a), b)$ becomes $s \cdot a \cdot b$. Given an ordered sequence of labels $\mathbf{a} = a_1, a_2, ..., a_k$ we will also use the shorthand $s \cdot \mathbf{a}$ to mean $S \cdot a_1 \cdot a_2 \cdot ... \cdot a_k$.

\begin{definition}[Commutativity]
Let $\sys = (D,L,U)$ be a semiautomaton, and $Q: D \rightarrow R$ be a query. We say that $U$ is {\em commutative} (resp. $Q$ is commutative) if for any state $s \in D$ and any sequences of labels $\mathbf{a},\mathbf{b}$ from $L$, we have $s \cdot \mathbf{a} \cdot \mathbf{b} = s \cdot \mathbf{b} \cdot \mathbf{a}$ (resp. $Q(s \cdot \mathbf{a} \cdot \mathbf{b}) = Q(s \cdot \mathbf{b} \cdot \mathbf{a})$).
\end{definition}

Commutative update implies commutative query, but the converse does not hold. A query may be indifferent to the order of updates but the state itself may be order-sensitive. For instance, the query counting the number of "a"s in a string with string concatenation as the update operation.

Insertion of tuples into a database instance is a commutative update. An update that allows insertion or deletion of tuples over a set-semantics database is not commutative; however insertion or deletion of tuples over a $\mathbb{Z}$-set semantics database is commutative (and forms an abelian group)~\cite{dbsp}.

\begin{proposition} \label{commutative-same-ft-value}
Let $\sys = (D,L,U)$ be a semiautomaton with a bottom state $\bot$, and $Q$ be a commutative query. Then, all free termination states return the same value for $Q$.
\end{proposition}

\begin{proof}
Take any two free termination states $s_1, s_2 \in D$. Let $\mathbf{a}_1$ be a sequence of updates such that $\bot \cdot \mathbf{a}_1 = s_1$ and let $\mathbf{a}_2$ be a sequence of updates such that $\bot \cdot \mathbf{a}_2 = s_2$. Because $Q$ is commutative, we have $Q(\bot \cdot \mathbf{a}_2 \cdot \mathbf{a}_1) = Q(\bot \cdot \mathbf{a}_1 \cdot \mathbf{a}_2) = r$. Since $s_1$ is a free termination state and $s_1 \rightarrow \bot \cdot \mathbf{a}_1 \cdot \mathbf{a}_2$, we have  $Q(s_1) = Q(\bot \cdot \mathbf{a}_1 \cdot \mathbf{a}_2) = r$. Similarly, $s_2$ is a free termination state and $s_2 \rightarrow \bot \cdot \mathbf{a}_2 \cdot \mathbf{a}_1$, so $Q(s_2) =r$. Thus, $Q(s_1) = Q(s_2)$.
\end{proof}

\begin{proposition} \label{commutative-always-reachable}
Let $\sys = (D,L,U)$ be a semiautomaton with a bottom state $\bot$, and $U$ be commutative.  If a query is freely terminating, then any state can reach a free termination state.
\end{proposition}

\begin{proof}
Let $s$ be the current state. Since $Q$ is freely terminating, there exists a free termination state $s_t$. Let $\mathbf{a}_t$ be an update sequence such that $\bot \cdot \mathbf{a}_t = s_t$, and $\mathbf{a}$ an update sequence such that $\bot \cdot \mathbf{a} = s$. From the commutativity of $U$, we have that $\bot \cdot \mathbf{a}_t \cdot \mathbf{a} = \bot \cdot \mathbf{a} \cdot \mathbf{a}_t = s'$. Since $s_t \rightarrow s'$, $s'$ must be a free termination state. But then we also have that $s \rightarrow s'$.
\end{proof}

\section{Examples of CRDTs}
\label{sec:crdt-examples}
\textit{State-based CRDTs} are sometimes referred to as the "semilattice data model" \footnote{This semilattice approach was also dubbed "ACID 2.0" by Helland, Campbell, and Finkelstein~\cite{buildingOnQuicksand}}.
A state-based CRDT must define three functions over its state, an \textit{update} operation that allows state to be mutated from outside the system, a \textit{merge} operation that determines how the states of two replicas can combine to converge to the same state, and a \textit{query} operation that defines what is visible to the user from the internal state~\cite{crdtsSurvey}. The basic idea is that updates from users can modify the state of any replica in the distributed system and asynchronous gossip of replica states in the background will ensure each replica eventually converges to the same state. In order to guarantee this convergence (formally known as \textit{strong eventual consistency}), the merge operation must be associative, commutative, and idempotent and the update operation must be inflationary with respect to the partial ordering induced by the merge operation. The intuition for why the merge must satisfy these three algebraic properties is that each one provides robustness to a different source of nondeterminism that can occur over asynchronous networks. Idempotence gives robustness to \textit{duplication} of messages which can occur via the network hardware itself or via retrying of message sends in the application. Commutativity gives robustness to \textit{reordering} of messages which can occur both in point-to-point streams of communication between machines and also in the order in which the same message is received when sent to multiple machines. Lastly, associativity gives robustness to the \textit{batching} of operations which allows states to be merged without respecting a left-fold operation order.

Below, we present some common examples of state-based CRDTs and discuss their algebraic properties and how they affect free termination.

\introparagraph{Grow-Only Set CRDT} This CRDT is defined as follows:
\begin{description}
    \item[state:] Set $S$ over the universe $\mathbb{U}$
    \item[update:] Set $S$ over $\mathbb{U}$ and  $m \in \mathbb{U}$. Then, $\textsf{update}(S, m) = S \cup \{m\}$ 
    \item[merge:] $\textsf{merge}(A, B) = A \cup B$
    \item[query:] $Q(S) = S$ (identity query)
\end{description}

   An append-only database is an example of a grow-only set. All results about free termination, threshold queries, and free termination of monotone and antitone queries apply to grow-only set CRDTs.
   
\introparagraph{Two-Phase Set CRDT}: We define it as follows:

\begin{description}
    \item[state:] Two sets over the universe $\mathbb{U}$, $\mathsf{INSERTS}$ and $\mathsf{DELETES}$.
    \item[update:] The operation "delete($m$)" gives $\mathsf{DELETES} = \mathsf{DELETES} \cup \{ m\}$. The operation "insert($m$)" gives $\mathsf{INSERTS} = \mathsf{INSERTS} \cup \{m\}$.
    \item[merge:] To merge two states, we set-union the $\mathsf{INSERTS}$ and $\mathsf{DELETES}$ sets respectively
    \item[query:] $\mathsf{INSERTS} - \mathsf{DELETES}$ (set difference)
\end{description}

This models something close to a database instance with insertions and deletions, but the semantics of deletion are slightly altered. The set of deletions is grow-only, so when an element is deleted once it is forever in the deletion set. This means that $m \in \mathsf{DELETES} \rightarrow m \notin \mathsf{INSERTS} - \mathsf{DELETES}$. In other words, we can negatively partially freely terminate on any value in $\mathsf{DELETES}$. On the other hand, we can never positively partially free terminate as any element in the $\mathsf\{INSERTS\}$ set could always be added to the $\mathsf{DELETES}$ set by some future update. In the paper Keep CALM and CRDT On ~\cite{crdtOn} this example was pointed out as a CRDT that has eventually consistent state, but a non-monotone user-observable query over that state. The paper stated the intuition that this CRDT offers weaker guarantees than a threshold-query over a CRDT would offer. We now see that intuition of weaker guarantees formalized by our definition of free termination. The grow-only set CRDT offers freely terminating threshold and dual-threshold queries. While the state of the two-phase set grows monotonically over time, its non-monotone query can only offer negative partial free termination.

Observe that the state transition graph of the Two-Phase Set CRDT is inflationary with respect to the set subset partial ordering, but if we look at an equivalent state transition graph with respect to the query, the graph contains cycles. The graph however is not fully invertible, which is what allows for negative partial free termination. 

\introparagraph{Grow-only Counter CRDT} This is defined as follows:

\begin{description}
    \item[state:] A  map $\mathsf{k}$ from unique replica IDs to natural numbers.
    \item[update:] An update at replica with ID $j$ increments the value of $\mathsf{k}[j]$.
    \item[merge:] Element-wise max of the map elements using the natural ordering of the natural numbers.
    \item[query:] Sum over the values in the map.
\end{description}

    The grow-only counter allows for eventually consistent counters for things like tallying votes or counting likes in a distributed setting. It is not sufficient to have the states be just a natural number that gets incremented because there is not an idempotent way to merge these states together, which may result in double-counting increments. To resolve this, each replica is assigned a unique ID and the only that replica will be able to apply updates that increment that unique ID's value. Effectively, the state of every node in the system is tracked in a map at each node and the merge operation overwrites this map with the more up-to-date values for any elements of the map your replica has out-of-date information on. $\mathsf{max}: \mathbb{N} \times \mathbb{N} \rightarrow \mathbb{N}$ is an associative, commutative, and idempotent operation so it forms a valid merge operation for a state-based CRDT. The update operation is monotone w.r.t. the ordering of $\mathsf{max}$, which means that it can act as a proxy for determining which states are \textit{newer} and which are outdated.

 \introparagraph{Positive-Negative Counter CRDT} We define this CRDT as follows:

\begin{description}
    \item[state:] A  map $\mathsf{k}$ from unique replica IDs to pairs of natural numbers.
    \item[update:] An update at replica with ID $j$ increments one of the two natural numbers in the pair at key $j$ in the map. An update meant to decrement the global counter will increment the right number and an update meant to increment the global counter will increment the left number.
    \item[merge:] Element-wise max of the map elements using the natural ordering of the pairs of natural numbers $(\mathbb{N} \times \mathbb{N})$.
    \item[query:] Compute a new map $\mathsf{k'}[j]= \mathsf{k}_1[j] - \mathsf{k}_2[j]$. Then sum over values in this map.
\end{description}

    The positive-negative-counter extends the grow-only counter to support decrementing counters in addition to incrementing them. The challenge beyond the lack of idempotence from the grow-only counter is that the obvious update operation is no longer monotone. To circumvent this, states are represented as monotonically growing pairs of natural numbers, with one representing the number of increments at that replica and the other representing the number of decrements at that replica.
    Observe that the range of the query is $\mathbb{Z}$ and all integer query results are reachable from all states in the semiautomaton. This means there is a fully invertible equivalent semiautomaton to the default CRDT representation, but the additional metadata kept around in the local states of the CRDT make the original semiautomaton fully acyclic and inflationary.

\section{Transducer Models with Metadata}
\label{sec:appendix-metadata}

We will show in this section how we can view a policy-aware transducer~\cite{weakerMonotonicity} via the lens of free termination. In the policy-aware setting, we are equipped with a {\em distribution policy} $P$ that maps each fact $t$ to a subset of the network nodes (which are the nodes that hold $t$). Each node can apply locally $P$ to any (possible) fact that uses constants from the current local active domain of the node.  

To do this, we will define a new semiautomaton $\sys_\cup^{\pm}$. Each state of the semiautomaton consists of a tuple of instances $(I^+, I^-)$ with the property that $I^+ \cap I^- = \emptyset$, i.e., the two instances have no tuples in common. Intuitively, $I^+$ tracks the presence of tuples, and $I^-$ tracks the absence of tuples. A transition in $\sys_\cup^{\pm}$ simply adds a tuple in either $I^+$ or $I^-$. The query result is computed by running the query on the positive instance, $q(I^+)$.

The distributed computational model is similar to the one presented in Section~\ref{sec:distributed-systems} with the following modifications. First, in a configuration $\gamma$ the buffer $buf$ is a tuple $(buf^+, buf^-)$ which tracks tuples that are to be ``added'' and tuples that are to be ``removed''. Second, the transitions are modified as follows:

\begin{description}
    \item[Produce Transition:] A node can move any tuple in its local state to the buffer $buf^+$ of its neighboring node; it can also check whether for a potential tuple $t$ (from the current local active domain) $P(t)$ contains that node and if not it can add $t$ to the buffer $buf^-$ of its neighboring node.
    \item[Consume Transition:] A node can update its state by removing a tuple from $buf^+$ and adding it to $I^+$, or removing a tuple from $buf^-$ and adding it to $I^-$.
    \item[Ready Transition:] A node can set $ready(v) \gets \mathsf{T}$.
\end{description}

A run is {\em fair} if every fact in every buffer is eventually taken out, and it is {\em complete} if every tuple in a state is eventually sent to its neighbors (and also every negative fact in a local policy is also sent to its neighbors). Correctness is defined in the same way as before.

For an instance $I$, we define $\bar{I}$ to be the set of tuples with values in $\textsf{adom}(I)$ that are not in $I$.

\begin{theorem} \label{thm:coord:free2}
$(Q,I)$ is coordination-free correct if and only if $(I,\bar{I})$ is a free termination state for $Q$ in the semiautomaton $\sys_\cup^{\pm}$.
\end{theorem}

\begin{proof} 
Consider the following algorithm: it will set $ready(v)$ to true exactly when $(I^+,I^-)$ is a free termination state for $Q$ in $\sys_\cup^{\pm}$. This algorithm is correct, since at the instance $(J^+,J^-)$ where $ready$ becomes true, all reachable states maintain the result of $Q$. Note that from the network flooding construction, $J^+ \subseteq I$ and $J^- \subseteq \bar{I}$. Hence, $(I,\bar{I})$ is reachable from $(J^+,J^-)$ and $Q(J^+) = Q(I)$. Further, every fair and complete run reaches the quiescence point, when $state(v) = (I, \bar{I})$, and thus the algorithm will set $ready$ to true.

In the other direction, suppose $(Q,I)$ is coordination-free correct and consider an algorithm that computes the output correctly. At the quiescence point, when $state(v) = (I, \bar{I})$, this algorithm must set $ready$ to true (since the state will remain unchanged from that point on). But now consider a fair and complete run for another input $I' \supseteq I$ such that $I'$ does not contain any tuples from $\bar{I}$, and assume that some node $v$ receives first all of $I, \bar{I}$. Note that $(I',\bar{I}')$ is a reachable state in $\sys_\cup^{\pm}$. At this point, the algorithm would need to do the ready transition. But because of correctness, it must be that $Q(I) = Q(I')$. Hence, $I$ is a free termination state.
\end{proof}

We say that a query $Q$ is {\em domain-distinct-monotone} if $Q(I) \subseteq Q(I \cup J)$ for all instances $I,J$ for which $J$ is domain distinct from $I$ (meaning every fact in $J$ has some constant that does not appear in $I$).
We can now show the following analogous theorem, which captures the characterization of policy-aware transducers in~\cite{weakerMonotonicity}.

\begin{theorem}
A Boolean query $Q$ is positively (resp. negatively) coordination-free in the policy-aware setting if and only if $Q$ is domain-distinct-monotone (resp. domain-distinct-monotone).
\end{theorem}

\begin{proof}
From Theorem~\ref{thm:coord:free2}, $Q$ is positively coordination-free if and only if every state $(I,\bar{I})$ with $Q(I) = \mathsf{T}$ is a free termination state. Take $J$ domain distinct from $I$ and suppose $Q(I)$ is true.
Then, $I$ is a free termination state. Since $(I \cup J, \bar{I \cup J})$ is reachable from $(I,\bar{I})$, $Q(I \cup J) = Q(I) = \mathsf{T}$. Hence, $Q$ is domain-distinct-monotone. The antitone case is symmetric.
\end{proof}

\end{document}